\documentclass[12pt]{article}
\usepackage{setspace}
\usepackage{amsmath,amsfonts,amsthm}
\usepackage{algorithmic}
\usepackage{algorithm}
\usepackage{array}
\usepackage{multicol,multirow}
\usepackage[caption=false,font=normalsize,labelfont=sf,textfont=sf]{subfig}
\usepackage{textcomp}
\usepackage{stfloats}
\usepackage{natbib}
\usepackage{url}
\usepackage{verbatim}
\usepackage{graphicx}
\newcommand{\Norm}[1]{\left\Vert#1\right\Vert}

\newcommand{\bd}{\boldsymbol d}

\newcommand{\bu}{\boldsymbol{u}}

\newcommand{\bx}{\boldsymbol{x}}
\newcommand{\balpha}{\boldsymbol{\alpha}}
\newcommand{\bbeta}{\boldsymbol{\beta}}

\def\Var{\mathop{\rm var}}

\newcommand{\bK}{\boldsymbol K}

\newcommand{\bX}{\boldsymbol X}

\newcommand{\bI}{\boldsymbol I}
\newcommand{\by}{\boldsymbol y}
\def\argmin{\mathop{\rm argmin}}
\def\hmi{\widehat{f}_0(\bx_i)}
\def\mi{{f}^*(\bx_i)}
\def\hpii{\widehat{\pi}_i}
\def\pii{{\pi}_i^*}

\newtheorem{lemma}{Lemma}

\newtheorem{theorem}{Theorem}

\newtheorem{prop}{Proposition}
\newtheorem{rem}{Remark}

\usepackage{amsmath,amsfonts,bm}

\def\eqref#1{equation~\ref{#1}}

\def\1{\bm{1}}

\DeclareMathAlphabet{\mathsfit}{\encodingdefault}{\sfdefault}{m}{sl}
\SetMathAlphabet{\mathsfit}{bold}{\encodingdefault}{\sfdefault}{bx}{n}

 \newcommand{\E}{\mathbb{E}}

\begin{document}

\title{Nonparametric augmented probability weighting with sparsity}

\author{
Xin He, Xiaojun Mao, and Zhonglei Wang
\thanks{X. He is with the School of Statistics and Management, Shanghai University of Finance and Economics, Shanghai 200433, China (e-mail:he.xin17@mail.shufe.edu.cn);
X. Mao is with the School of Mathematical Sciences, Shanghai Jiao Tong University, Shanghai 200240, China (e-mail: maoxj@sjtu.edu.cn). Z. Wang is with Wang Yanan Institute for Studies in Economics and School of Economics, Xiamen University, Xiamen, Fujian 361005, China  (e-mail: wangzl@xmu.edu.cn). 
}
\thanks{X. He and X. Mao contributed equally to this work. Corresponding authors: Z. Wang.} 
}



\maketitle

\begin{abstract}
Nonresponse frequently arises in practice, and simply ignoring it may lead to erroneous inference. Besides, the number of collected covariates may increase as the sample size in modern statistics, so parametric imputation or propensity score weighting usually leads to inefficiency without consideration of sparsity. In this paper, we propose a nonparametric imputation method with sparse learning by employing an efficient kernel-based learning gradient algorithm to identify truly informative covariates. Moreover, an augmented probability weighting framework is adopted to improve the estimation efficiency of the nonparametric imputation method and establish the limiting distribution of the corresponding estimator under regularity assumptions. The performance of the proposed method is also supported by several simulated examples and one real-life analysis.


\end{abstract}

\newpage
\doublespacing
	\section{Introduction}

Nonresponse is a common problem in social science and other related fields, and simply ignoring it may lead to inefficiency or even erroneous inference due to confounding covariates \citep{rosenbaum1983central,qu2010highly,abadie2016matching, LIN2018}. Moreover, the  number of collected covariates is relatively large in modern statistics, which makes learning nonresponse even more challenging \citep{Yang2020}. How to deal with the nonresponse under the high-dimensional setup is still an open question. 

Sparse learning bridges the gap between the  high-dimensional data analysis and nonresponse. It is generally believed that among the numerous covariates, only a few  of them contribute to the response of interest, known as truly informative ones, while others are noise. Thus, a variety of sparse learning methods have been proposed to identify those truly informative covariates under regularity assumptions.  The linear response model assumption is popularly imposed, and various attempts have been made  by designing sparse-induced regularization  \citep{Tibshirani1996, Fan2001,ZouH2006, Shen2012, Shen2013},  evaluating the marginal dependence \citep{FAN2008, WangX2016}, or checking variable robustness against added noise \citep{Barber2018}. Extended methods have also been developed for nonparametric models \citep{Lin2006, Huang2010, Fan2011}. However, all these methods require explicit model assumptions that are difficult to validate in practice or suffer from heavy computational burden. To circumvent this difficulty, \citet{Belloni2013} proposed a feasible lasso method, which is similar to the adaptive lasso \citep{ZouH2006}, for variable selection in a partial linear model. By using machine learning algorithms to handle high-dimensional nuisance parameters, \citet{Chernozhukov2018} proposed a double/debiased machine learning procedure to achieve parametric convergence rate for a low dimensional parameter.  A valid double robust estimator using lasso-type penalty is discussed by \citet{Tan2020}. Recently, kernel-based sparse learning methods have been inspired by the fact that the gradient functions provide an appropriate criterion to identify a general dependence structure in a model-free fashion. Specifically, \citet{Rosasco2013} proposed a novel learning-gradient method, which adds an empirical functional penalty on the gradients to a standard kernel ridge regression in a reproducing kernel Hilbert space (RKHS). Besides, \citet{YangL2016} employed  pair-wise learning to estimate the gradient functions and considered a functional group lasso penalty to induce sparsity.  It is worth pointing out that  the lack of  selection consistency  \citep{Rosasco2013} and the high computational cost \citep{YangL2016} remain unsolved. To alleviate the difficulties, \citet{He20192} proposed a two-step sparse learning framework, which is computationally efficient in the sense that it only requires to fit the standard kernel ridge regression and the selection consistency is established under regularity assumptions. The method proposed by \citet{He20192} can  be regarded as a nonparametric joint screening approach and achieves methodological flexibility, numerical efficiency and asymptotic consistency simultaneously.

Propensity score weighting is commonly used to handle nonresponse \citep{aipw1994,WOOLDRIDGE20071281,tan2010bounded,graham2012inverse,zhao2017semiparametric}, but  conventional methods using all covariates may lead to numerical failure, including the lack of convergence and inefficiency, due to overfitting. Thus, sparse assumption is often imposed to estimate the propensity scores more efficiently \citep{Shevada2003,Genkin2007}. A Bayesian variable selection method has been proposed by \citet{Chen1999} for logistic regression; also see \citet{wainwright2007high}, \citet{banerjee2008model} and \citet{10.1214/09-AOS691} for details about penalized logistic regression models. The  group lasso  \citep{yuan2006model} was generalized to logistic regression model by  \citet{Kim2006}, and \citet{Meier2008}  proposed a more efficient group lasso algorithm than that of \citet{Kim2006}. Besides, \citet{Meier2008} also established the asymptotic consistency of the corresponding estimator.  \citet{Ning2020} proposed a high-dimensional covariate balancing propensity score estimator, and they validated that their proposed estimator is of parametric convergence rate and is asymptotically normal distributed. See \citet{tang2014feature} and \citet{10.2307/26408297} for a review of identifying informative covariates for logistic regression models.

In this paper, we propose a nonparametric Augmented Inverse Probability Weighting (AIPW) framework \citep{aipw1994} to handle the nonresponse under the assumption of sparsity. Inspired by the key observation that the gradient functions provide an appropriate information of the truly informative covariates in a model-free fashion, we employ a kernel-based sparse learning algorithm to efficiently impute the nonresponse. It only requires to fit the standard kernel ridge regression, which has an analytical solution, and the gradient functions can be directly computed by the derivative reproducing property. More importantly,  the truly informative covariates can be exactly recovered with high probability. Even though  the nonparametric  imputation with sparse learning achieves consistency, its convergence rate is at most $O_{\mathbb{P}}(n^{-1/6}\log(n))$ under regularity assumptions, so it is hard to construct an interval estimator. To alleviate this difficulty, an AIPW framework is adopted to  improve the convergence rate of the corresponding estimator, and a central limit theorem can be established. To achieve this goal, certain  propensity score methods for analyzing sparse data suffice under regularity conditions; see Section~\ref{sec: group lasso} for details. The corresponding variance estimator is also discussed.   The superior performance of the proposed nonparametric AIPW framework is also supported by the numerical comparisons against some state-of-the-art methods in several simulated examples and one real-life analysis.

The rest of this paper is organized as follows. Section \ref{gen_inst} provides the background and introduces the proposed nonparametric AIPW framework.  The theoretical properties of the corresponding estimator are established under  regularity assumptions in Section \ref{headings}. Section \ref{others} reports the numerical experiments on the simulated and real-life examples. A brief summary is provided in Section \ref{Conclude}.

\section{Method}
\label{gen_inst}\label{sec: feature selection by derivatives}

Consider
\begin{align}\label{true:model}
	Y= f^*(\bx)+\epsilon,
\end{align}
where $f^*(\bx) = \E(Y\mid \bX=\bx)$ is a continuous function of a covariate vector $\bx=(x_1,...,x_p)^{\top}$ taking values from  a $p$ dimensional separable and compact metric space ${\cal X}\subset \mathbb{R}^p$,  and $\epsilon$ denotes a random noise   with conditional mean zero and bounded variance.  We are interested in inferring $\theta^* = \E(Y)$ from a random sample $\{(\bx_i,y_i):i=1,\ldots,n\}$ generated by (\ref{true:model}).

If $y_1,\ldots,y_n$ were fully observed, the  sample mean $\widehat{\theta}_n=n^{-1}\sum_{i=1}^ny_i$ would be an efficient estimator of $\theta^*$. However, it is generally not the case in practice, and the response of interest  suffers from nonresponse. For $i=1,\ldots,n$, denote $\delta_i$ to be the response indicator of $y_i$, where $\delta_i=1$ if $y_i$ is observed and 0 otherwise. For simplicity, we assume missing at random  \citep{rubin1976inference}  for the response mechanism, 
\begin{equation}\label{eq: MAR 1}\Pr(\delta_i=1\mid \bx_i,y_i) = \Pr(\delta_i=1\mid \bx_i),\end{equation} and denote $\pi^*(\bx)=\Pr(\delta=1\mid \bx)$.

If consistent estimators $\widehat{f}_0(\bx)$ and $\widehat{\pi}(\bx)$ for $f^*(\bx)$ and $\pi^*(\bx)$ are available, then an AIPW estimator,
\begin{equation}
	\widehat{\theta}_{AIPW} = \frac{1}{n}\sum_{i=1}^n\left[\widehat{f}_0(\bx_i) + \frac{\delta_i}{\widehat{\pi}(\bx_i)} \{y_i-\widehat{f}_0(\bx_i)\} \right],\label{eq: final est}
\end{equation}
can be applied to estimate $\theta^*$. More rigorously, the estimator (\ref{eq: final est}) is not an AIPW estimator unless we replace $n^{-1}$ by $(\sum_{i=1}^n\delta_i\widehat{\pi}_i^{-1})^{-1}$. However, if the response model is correctly specified, we can show that $n^{-1}\sum_{i=1}^n\delta_i\widehat{\pi}_i^{-1}\to1$ in probability under regularity conditions. When $p$ is small and $f^*(\bx)$ is a parametric model, standard statistical methods can be used to obtain  $\widehat{f}_0(\bx)$ and $\widehat{\pi}(\bx)$. As $p$ increases, however, it is not reasonable to include all  covariates to estimate $f^*(\bx)$ and $\pi^*(\bx)$ due to the curse of dimensionality. Moreover, the model misspecification for $f^*(\bx)$ leads to erroneous inference.

\subsection{Estimation of ${f}^*(x)$ via nonparametric sparse learning}
To overcome those difficulties, we employ an efficient kernel-based sparse learning algorithm \citep{He20192} to estimate  $f^*(\bx)$ in (\ref{true:model}). Denote ${\cal H}_K$ to be an RKHS induced by a pre-specified kernel function $K(\cdot,\cdot)$, where  $K(\cdot,\cdot):{\cal X}\times{\cal X}\rightarrow {\mathbb{R}}$ is bounded, symmetric and positive semi-definite. It can be shown that ${\cal H}_K$ associated with the kernel $K(\cdot,\cdot)$ is the completion of the linear space spanned by $\{K_{\bx}(\cdot): \bx {\in} {\cal X} \}$ with an inner product  $\langle K_{\bx}, K_{\bu}\rangle_K=K(\bx, \bu)$ for $\bx, \bu \in {\cal X}$, where $K_{\bx}(\cdot)=K(\bx,\cdot)$. Thus, ${\cal H}_K$ is uniquely determined by a kernel function  $K(\cdot,\cdot)$ and the reproducing property, $\langle f, {K}_{\bx} \rangle_{K}=f(\bx)$ for  $f \in {\cal H}_K$ and $\bx\in{\cal X}$. 
It is noteworthy that the RKHS induced by some universal kernel, such as the Gaussian kernel, is  fairly large  in the sense that any continuous function can be  well approximated \citep{Steinwart2005}. Thus, by (\ref{eq: MAR 1}), we assume
\begin{align}\label{eqn:motivate21}
	f^*(\bx)=\argmin_{f \in {\cal H}_K} \E  [\delta\{y-f({\bx})\}]^2.
\end{align}
By (\ref{eq: MAR 1}), (\ref{eqn:motivate21}) implies that $f^*(\bx)= \E(Y\mid \bX=\bx)\in {\cal H}_K$.

A covariate $x_l$ is non-informative for  $f^*(\bx)$, if and only if $g^*_{l}(\bx)=0$ for ${\bx} \in {\cal X}$ almost surely, where  $x_l$ is the $l$th component of $\bx$ for $l=1,\ldots,p$, and $g_l^*(\bx)=\partial f^*(\bx)/\partial x_l$.
Thus, the usefulness of $x_l$ for estimating $f^*(\bx)$ can be measured via the  ${\cal L}^2$-norm of $g^*_{l}(\bx)$, 
$
\| g^*_{l} \|^2_{2} =\int_{\cal X}  \{ g_{l}^*(\bx) \}^2 d\rho(\bx),
$
where $\rho(\bx)$ is the marginal distribution of $\bX$. Denote ${\cal A}^*=\big \{l:\| g^*_{l} \|^2_{2}>0 \big \}$ to be the active set containing all informative covariates associated with  $f^*(\bx)$. To estimate $g^*_l(\bx)$ efficiently, we consider the following derivative reproducing property \citep{ZhouDX2007}:
\begin{align}\label{eqn:motivate2}
	g^*_{l}(\bx)= \langle f^*, {\partial_l {K}_{\bx}} \rangle_K,
\end{align}
where $\partial_l {K}_{\bx}(\cdot) = \partial K(\bx,\cdot)/\partial x_l$. Specifically, by (\ref{eqn:motivate21})--(\ref{eqn:motivate2}),  once an initial estimator of  $f^*(\bx)$ is available, say $\widehat{f}(\bx)$, its gradient function $g^*_{l}(\bx)$ can be estimated by (\ref{eqn:motivate2}). Without loss of generality, we assume that the first $m$ samples are fully observed, and the subsequent analysis is conditional on the realised sample, where $m=\sum_{i=1}^n\delta_i$. Under (\ref{eq: MAR 1}) and the high-dimensional setup, to obtain  an initial estimator  $\widehat{f}(\bx)$,  we  employ the standard  kernel ridge regression  by solving  
\begin{equation}\label{eqn:distmin}
	\widehat{f}=\underset{f\in{\cal H}_K}{\arg\min}  \Big [ \frac{1}{m}\sum_{i=1}^m \left\{y_i- f(\bx_i)\right\}^2 +\lambda\Norm{f}_{K}^2 \Big ],
\end{equation}
where   $\lambda>0$ is a tuning parameter controlling the model complexity and typically goes to 0 as $m$ goes to infinity, and  $\Norm{\cdot}_{K}^2$ denotes the RKHS-norm induced by the inner product $\langle \cdot, \cdot \rangle_K$; see Section~\ref{headings} for details.  By the representer theorem \citep{Mercer1909}, the solution of (\ref{eqn:distmin}) is of the form
\begin{align}\label{eqn:22}
	\widehat{f}(\bx)=\sum_{i =1}^m \widehat{\alpha}_iK({\bx}_i,\bx)=\widehat{\balpha}^\top{\bK}_{m}(\bx),
\end{align}
where ${\bK}_{m}(\bx)=(K({\bx}_1,\bx),...,K({\bx}_m,\bx))^\top$, and $\widehat{\balpha}=(\widehat{\alpha}_1,...,\widehat{\alpha}_m)^\top\subset \mathbb{R}^m$ are the estimated representer coefficients.  That is,  the representer theorem converts  the original optimization  task (\ref{eqn:distmin}) in  an infinite functional space ${\cal H}_K$ into a $m$-dimensional vector space. By  (\ref{eqn:22}),  the optimization task (\ref{eqn:distmin}) is equivalent to 
\begin{equation*}\label{eqn:distmin2}
	\widehat{\balpha}=\underset{\balpha}{\argmin}  \Big [ \frac{1}{m}\sum_{i=1}^m  \left\{y_i- \balpha^\top{\bK}_{m}(\bx_i)\right\}^2 +\lambda\balpha^T\bK \balpha\Big ],
\end{equation*}
and its solution is
$
\widehat{\balpha}= (\bK  +\lambda \bI  )^{-1}  {\by},
$
where $\bK$ is an $m\times m$ matrix with $(i,j)$th entry being $K(\bx_i,\bx_j)$  and $\by =(y_1,\ldots,y_m)^{\top} \in \mathbb{R}^m$.

Once $\widehat{\balpha}$ is obtained, the gradient function in (\ref{eqn:motivate2}) can be estimated  by
$$
\widehat{g}_{l}(\bx)= \frac{\partial \widehat{f}({\bx})}{\partial x_l } = \widehat{\balpha}^\top  {\partial_l {\bK}_m({\bx})}, \quad (l=1,...,p),
$$
where  ${\partial_l {\bK}_m({\bx})}=(\partial_l {K}_{\bx_1}({\bx}),\ldots, \partial_l {K}_{\bx_m}({\bx}))^\top$. Since  the marginal distribution $\rho(\bx)$ is seldom available, instead of the  ${\cal L}^2$-norm, the empirical norm $\Norm{\cdot}_m$ is considered:
$$
\| \widehat{g}_{l}  \|^2_m= \frac{1}{ m} \sum_{i=1}^m \big \{\widehat{g}_{l}(\bx_i) \big \}^2=  \frac{1}{ m}  \sum_{i=1}^m  \big \{ \widehat{\balpha}^\top  {\partial_l {\bK}_m({\bx}_i)} \big \}^2,
$$
and  the estimated active set is
$
\widehat{\cal A}_{v_m}=\Big \{l: \| \widehat{g}_{l}  \|^2_m > v_m \Big \},
$
where  $v_m$ is a  thresholding value determined through a stability-based selection criterion \citep{SunWW2013}. Finally, we refit (\ref{eqn:distmin}) with the selected covariates in    $\widehat{\cal A}_{v_n}$ to obtain the nonparametric estimator $\widehat{f}_0(\bx)$.

It is worthy pointing out that the  employed  sparse learning algorithm  was originally proposed by \citet{He20192}, and they only focused on the purpose of variable selection and established the selection consistency without considering nonresponse. Yet, we generalized their method to handle incomplete samples in this paper and treat it as an valid intermediate  estimator of our proposed  estimator.  More importantly, we further established a central limit theorem for the proposed nonparamtric estimator, which is rare and attractive in machine learning, and a variance estimation is also provided as well; see Section~\ref{headings} for details. 

\subsection{Estimation of $\pi^*(x)$}\label{sec: group lasso}
There exist works to estimate $\pi^*(\bx)$ under the assumption of sparsity, and we consider the group lasso \citep{Meier2008} as an example  by assuming
\begin{equation*}
	\mbox{logit}\{\pi^*(\bx_i)\} = \beta^*_0 +\bx_i^\top\bbeta^*_1,\label{eq: response prob} 
\end{equation*}
where $\mbox{logit}(z) = \log(z) - \log(1-z)$ for $z\in(0,1)$. In addition, assume that the covariate vector can be rewritten as $\bx = (\bx_{1}^{\top},\ldots,\bx_{G}^{\top})^{\top}$, where $\bx_{g}\in\mathbb{R}^{\mbox{\scriptsize df}_g}$ contains the covariates of the $g$th group for $g=1,\ldots,G$, and $\mbox{df}_g$ is the corresponding degrees of freedom. For example, $\mbox{df}_g=3$ if $\bx_{g}$ corresponds to a categorical covariate with four levels, and $\mbox{df}_g=1$ if $x_{g}$ is continuous; see \citet{Meier2008} for details.

The log-likelihood  estimator with a group lasso penalty  is obtained by solving
\begin{eqnarray}
	\widehat{\bbeta}_{\lambda_2}
	&=&\argmin_{\bbeta}  \Big \{ -l(\bbeta) +\lambda_2p(\bbeta)    \Big \}, \label{eq: logsitic28}
\end{eqnarray}
where $l(\bbeta) = \sum_{i=1}^n[\delta_i\log\{\pi(\bx_i)\} + (1-\delta_i)\log\{1-\pi(\bx_i)\}]$ is the log-likelihood of the response indicators, $\bbeta^\top=(\beta_0,\bbeta_1^\top) \in\mathbb{R}^{p+1}$ with $\beta_0\in\mathbb{R}$ and $\bbeta_1\in\mathbb{R}^p$, $p(\bbeta) =\sum_{g=1}^G \mbox{df}_g^{1/2}\lVert \bbeta_g\rVert_2  $ is the group lasso penalty, $\lVert\cdot\rVert_2$ is the Euclidean norm, and $\bbeta_g$ corresponds to  $\bx_{i,g}$ for $g=1,\ldots,G$.  The block co-ordinate gradient descent algorithm  is used to obtain $\widehat{\bbeta}_{\lambda_2}$ in (\ref{eq: logsitic28}), and the detailed algorithm is adjourned to Appendix~\ref{append: BCGD ALG}.

\begin{rem}
	Since the estimated response probability is used to improve the convergence rate of the estimator in (\ref{eq: final est}), the response model is assumed to be correctly specified; see \citet{Shenkan2017} for a similar assumption. In addition to the group lasso method \citep{Meier2008}, other penalized logistic regression estimators \citep{fan2014strong,Ning2020} can be used to estimate the response probability. However, to guarantee the asymptotic central limit theorem in Theorem~\ref{thm:clt}, the estimated response probability by other methods should satisfy Lemma~\ref{cor: 1}; see Section~\ref{headings} for details. 
\end{rem}

\section{Theoretical Properties}
\label{headings}

In this section, we investigate the asymptotic consistency of $\widehat{f}_0(\bx)$ and establish the central limit theorem for the  AIPW estimator    in (\ref{eq: final est}) under  regularity assumptions.

Denote an integral operator  
$L_K: {\cal L}^2({\cal X},{\rho}) \rightarrow  {\cal L}^2({\cal X},{\rho})$ as
$
L_K(f)(\bx)=\int_{\cal X} K(\bx,\bu)f(\bu)d\rho(\bu),
$
for  $f \in {\cal L}^2({\cal X},{\rho})$, where ${\cal L}^2({\cal X},{\rho})=\{f: \int f^2(\bx)d\rho(\bx)<\infty\}$.  If the RKHS ${\cal H}_K$ is separable, then by the spectral theorem \citep{Fischer2020}, we have
$
L_Kf=\sum_{j\geq 1}\mu_j\langle f , e_j \rangle_2e_j,
$
where $\{e_j:j=1,2,\ldots\}$ form an orthonormal basis of ${\cal L}^2({\cal X}, \rho)$, $\{\mu_j:j=1,2,\ldots\}$ are the corresponding eigenvalues with respect to $L_K$, and $\langle f, g\rangle_2 = \int_{\cal X} f(\bx)g(\bx)d\rho(\bx)$ denotes the inner product of $f(\bx)$ and $g(\bx)$ in ${\cal L}^2({\cal X},\rho)$.  
By Mercer's theorem \citep{Steinwart2008}, under regularity assumptions, the eigen-expansion of $K(\bx,\bu)$ is
$
K(\bx,\bu)=\sum_{j\geq 1}\mu_j e_j(\bx)e_j(\bu).
$
Hence, the RKHS-norm of any  $f \in {\cal H}_K$ can also be expressed as
\begin{align*}
	\|f\|^2_K= \sum_{j\geq 1} \frac{ \langle f, e_j \rangle^2_2}{\mu_j}.
\end{align*}
The above result implies the decay rate of $\mu_j$ fully characterizes the complexity of the RKHS. 

The following technical assumptions are made to investigate the theoretical properties of  the AIPW in (\ref{eq: final est}).

{\noindent \bf Assumption 1}: There exists a positive constant $r\in (1/2,1]$, such that $f^*(\bx)$ is in the range of the $r$th power of $L_K$, denoted as $L_{K}^r$.  Besides, the distribution of  $\epsilon$ has a $q$-exponential tail with some function $q(\cdot)$; that is, there exists a constant $c_1>0$, such that $\Pr(\lvert\epsilon\rvert>t)\leq c_1\exp\{-q(t)\}$ for any $t>0$.

{\noindent \bf Assumption 2}: There exist positive values  $\kappa_{1,p}$ and $\kappa_{2,p}$, which may depend on $p$, such that $\sup\limits_{\bx \in {\cal X}} \|K_{\bx}\|_K\leq \kappa_{1,p}$ and $\sup\limits_{\bx \in {\cal X}}\|\partial_lK_{\bx}\|_K\leq \kappa_{2,p}$ for  $l=1,...,p$.

{\noindent\bf Assumption 3}: There exists a positive constant $\xi_1 >1$ such that
\begin{align*}
	&\min_{l \in {\cal A}^*}\left \| g^*_l \right \|^2_{2}> \\
	&c_{m} \max \left\{ \kappa_{1,p}\|f^*\|_K, q^{-1}\left(\log\frac{4c_1 m}{\delta_m}\right) \right\} m^{-\frac{2r-1}{2(2r+1)}} (\log p)^{\xi_1},
\end{align*}
where $c_m$ is provided in Lemma \ref{thm1}.

{\noindent \bf Assumption 4}: There exists $\kappa\in(0,1/2)$ such that $\kappa < \pi(\bx)<1-\kappa$ for all $\bx\in{\cal X}$.

{\noindent \bf Assumption 5}: $\E(\bX\bX^\top)$ is invertible, and its smallest eigenvalue is bounded away from zero by a fixed positive constant $c_{min}$, and recall that $\bX$ is the random vector associated with $\bx_1,\ldots,\bx_n$.

{\noindent \bf Assumption 6}: Let $\bX_g$ be the random vector associated with the $g$th group, and we normalize $\bX_g$ such that $\E(\bX_g^{\top}\bX_g)$ is a $\mbox{df}_g\times\mbox{df}_g$ identity matrix. Then, there exists $L_n$ such that 
$
\max_{\bx}\max_{g}(\bx_g^{\top}\bx_g)\leq nL_n^2,
$
where $\bx_g$ corresponds to the normalized $\bX_g$.

{\noindent \bf Assumption 7}: $\max_{g=1,\ldots,G}\mbox{df}_g=O(1)$, there exists a constant number $\zeta>0$ such that $\log(G)=o(n^{1/3-2\zeta})$ and $G\gg \log(n)$,  $N_0=O(1)$, $\lambda_2 \asymp \log(G)$, i.e. $\lambda_2$ is of the order $\log(G)$, and $L_n^2 = O\{1/\log(G)\}$, where $N_0$ is the number of non-zero group effects. 

Assumptions 1--3 are proposed for the kernel-based sparse learning algorithm, and Assumptions 4--7 are required by the group lasso logistic regression. In Assumption 1, the integral operator $L_K$ is self-adjoint and semi-positive definite, so its fractional operator $L_{K}^r$ is well-defined, and its range is contained in ${\cal H}_K$ as long as $r\geq 1/2$; see \citet{Smale2007} and \citet{Mendelson2010} for details. This implies that for some function $h \in  {\cal L}^2({\cal X}, \rho)$, it holds
$
L^r_Kf^*=\sum_{j\geq 1}\mu^r_j\langle h , e_j \rangle_2e_j \in {\cal H}_K,
$ 
so ensures strong estimation consistency under the RKHS-norm.  The second part of Assumption 1 characterizes the tail behavior of the error distribution, and it relaxes the commonly-used bounded response assumption in the machine learning literature \citep{Smale2007, Rosasco2013, Lv2018}. Besides, the assumption on the error distribution is general and can be satisfied by a variety of distributions \citep{WangX2016, ZhangC2016}. For example, if the error distribution  is sub-Gaussian or bounded, then  $q(t)=O(t^2)$ suffices; if the distribution of $\epsilon$ is sub-exponential,  $q(t)=O(\min\{t/C, t^2/C^2 \})$ suffices with $C>0$.  Assumption 2 bounds  the RKHS-norms associated with the kernel function and its gradient functions, and it is satisfied by popular kernels, including the Gaussian kernel, linear kernel and the Sobolev kernel \citep{Smale2007, Rosasco2013, YangL2016}. For example, $\kappa_{1,p}=\kappa_{2,p}=1$ for the Gaussian kernel, $K(\bx, \bu)=\exp\{- \|\bx-\bu \|_2^2 /(2\sigma^2)\}$, and  $\kappa_{1,p}=Cp$ and $\kappa_{2,p}=C$ for the linear kernel, $K(\bx, \bu)=\bx^\top\bu$, for some positive constant $C$. Assumption 3 requires that  the  gradient functions contain sufficient information about the truly informative covariates. It is worthy pointing out that we measure the significance  of each gradient function  to distinguish  informative and uninformative covariates without any explicit model specification. The  minimal signal strength in Assumption 3 is much tighter than those in other nonparametric sparse learning methods \citep{Huang2010, YangL2016}, which often require the signal strength to be  bounded below by some positive constant.
Assumption~4 bounds the response probability, and it is commonly used to avoid inefficient estimators. To obtain the desired convergence rate, Assumption~4  guarantees that $m\asymp n$ in probability, where $a_n\asymp b_n$ is equivalent to $a_n=O(b_n)$ and $b_n=O(a_n)$.
The smallest eigenvalue of  $\E(\bX\bX^\top)$ is bounded by a fixed positive constant in Assumption~5, and it is a special case of  assumption (b) of \citet{Meier2008}. In Assumption~6, the convergence rate of $L_n$ is related with that of the estimated response probability. Assumption~7 is used to guarantee that 
\begin{equation}\label{eq: append 1}
	\E\{\lvert\eta_{\widehat{\bbeta}_{\lambda_2}}(\bX) - \eta_{{\bbeta}^*}(\bX) \rvert^2\} =  O_{\mathbb{P}}(n^{-2/3-2\zeta}),
\end{equation}
where $\eta_{\bbeta}(\bx)=\beta_0 +\bx^\top\bbeta_1$, $\bbeta^* = (\beta_0^*,\bbeta_1^*)$ and the expectation is taken with respect to $\bX$ conditional on $\widehat{\bbeta}_{\lambda_2}$ or the observations. Specifically, the value $\zeta$ is used to show $n^{-\zeta/2}\log(n)\to0$ as $n\to 0$, so it can be chosen arbitrarily small; see the proof of Theorem~\ref{thm:clt} for details.

\begin{lemma}\label{thm1}
	Suppose  Assumptions 1--2 are satisfied, 
	and $\lambda={m^{-1/(2r+1)}}$. Then, for any $\delta_m \in (0,1)$, with probability at least $1-\delta_m$, there holds
	\begin{align}\label{thm13}
		\max_{1 \leq l \leq p} \ \big\lvert  \| \widehat{g}_l \|^2_m -  \| g^*_l \|_{2}^2 \big\rvert  \leq c_{m} c_{p, q}  {\log\Big(  \frac{8p}{\delta_m}\Big)}{m^{-\Theta}},\notag
	\end{align}
	where 
	$c_{m}$ is a constant depending only on $\kappa_{1,p}, \kappa_{2,p}$ and $\|f^*\|^2_{K}$, $c_{p, q}=\max \left\{ \kappa_{1,p}\|f^*\|_K, q^{-1}\left(\log\frac{4c_1 m}{\delta_m}\right) \right\}$ with $q^{-1}(\cdot)$ denoting the inverse function of $q(\cdot)$, and $\Theta={\frac{2r-1}{2(2r+1)}}$.
\end{lemma}

Lemma \ref{thm1}  guarantees that $\left \|\widehat{g}_l\right \|^2_m$ converges to $\| g_l^* \|_{2}^2$ with high probability, and it is crucial to establish the selection consistency of the employed sparse learning algorithm. The convergence result in Lemma~\ref{thm1} still holds even when the dimension diverges with the sample size, and the quantities $\|f^*\|^2_{K}$ and  $\|L^{-r}_{K}f^*\|_2$, which  may depend on the number of truly informative covariates of $f^*(\bx)$,  may also diverge as the sample size increases. For instance,  if $f^*(\bx)=\bx^T\bbeta^*$, then $\|f^*\|^2_{K}=\|\bbeta^*\|_2^2$, which clearly depends on the number of truly informative covariates. However,  such dependency is difficult to quantify explicitly in a fully general case \citep{Fukumiza2014}.

The following lemma establishes the asymptotic selection consistency of the proposed sparse learning method.

\begin{lemma}\label{thm3}
	Suppose that  the assumptions of Lemma \ref{thm1} and Assumption 3 are satisfied. If $v_m=0.5c_m c_{p,q} m^{-\Theta} (\log p)^{\xi_1},$ then
	$
	\Pr\big( \widehat{\cal A}_{v_m}={\cal A}^* \big) \rightarrow 1, \ \ \mbox{as} \ \ m\rightarrow \infty.
	$
\end{lemma}

Lemma \ref{thm3} shows that the selected covariates can exactly recover the truly informative ones with probability tending to 1. This result is particularly general in that it is established  without any   model specification.  The following theorem shows that  $\widehat{f}_0(\bx)$ achieves a fast convergence rate in term of the infinity norm, where $\widehat{f}_0(\bx)$ is obtained by the  standard kernel ridge regression (\ref{eqn:distmin}) based on the selected covariates in $\widehat{\cal A}_{v_n}$.

\begin{theorem}\label{cor: cor 1}
	Suppose the assumptions of Lemma \ref{thm3} are satisfied and denote the probability $\Pr(\widehat{\cal A}_{v_n}\neq {\cal A}^*)= \Delta_m$. If $\lambda={m^{-\frac{1}{2r+1}}}$, then with probability at least $1-\delta_m-\Delta_m$, there holds
	\begin{eqnarray}
		\|\widehat{f}_0 -f^* \|_{K}   &\leq& c_{m, 2} c_{p_0,q}  {\log\Big(  \frac{4}{\delta_m}\Big)} { m^{-\Theta}  } ,\notag
	\end{eqnarray}
	where $p_0={\vert}{\cal A}^*{\vert}$, 
	\begin{eqnarray}
		c_{m, 2}&=&4\max\{\kappa^2_{2,p_0}, \kappa^2_{2,p_0}\|f^*\|_{K}, \|f^*\|^2_{K}\}\notag \\ 
		&&\times\max\{3\kappa_{1,p_0},  2\sqrt{2}\kappa_{2,p_0}^2, \|L_K^{-r}f^*\|_2\},\notag
	\end{eqnarray} 
	and $c_{p_0, q} =\max \left\{ \kappa_{1,p_0}\|f^*\|_K, q^{-1}\left(\log\frac{4c_1 m}{\delta_m}\right) \right\}$.\\
	Additionally, if we take $r=1$, and assume that $p_0=O(1)$, $\epsilon$ has sub-Gaussian or sub-exponential tail and  by Assumption 4, we have
	$
	\left \|\widehat{f}_0 -f^* \right \|_{\infty}   =O_{\mathbb{P}}(  n^{-\frac{1}{6}}\log n  ).
	$
	
\end{theorem}
Theorem \ref{cor: cor 1} establishes the  uniform  convergence rate of the refitted estimator $\widehat{f}_0(\bx)$, and it plays a crucial role to establish the central limit theory of the AIPW estimator in (\ref{eq: final est}). The required tail behavior of $\epsilon$ in Theorem \ref{cor: cor 1} is to quantify $q^{-1}(\cdot)$ explicitly for simplicity, and it can be extended to any error distribution satisfying Assumption 1.

\begin{lemma}\label{cor: 1}
	Suppose  Assumptions 4--7 are satisfied. Then, given $\widehat{\bbeta}_{\lambda_2}$, there holds
	$
	\E[\lvert\widehat{\pi}(\bX) - \pi^*(\bX)\rvert^2] = O_{\mathbb{P}}(n^{-2/3-2\zeta}).
	$
\end{lemma}
Lemma~\ref{cor: 1} establishes the convergence rate of the estimated response probability using group lasso logistic regression. By Lemma~\ref{cor: 1}, we essentially require that the estimated response probability should be at least consistent. A similar requirement is also discussed by \citet{Tan2020}. Specifically, \citet{Tan2020} assumed a correctly specified response model in order to achieve valid interval estimator. If other penalized logistic regression estimators are considered, Assumptions~4--7 should be replaced in order to guarantee Lemma~\ref{cor: 1}.

By Theorem~\ref{cor: cor 1} and Lemma~\ref{cor: 1}, we can  validate the following central limit theorem for the AIPW estimator in (\ref{eq: final est}). 
\begin{theorem}\label{thm:clt}
	Suppose all the assumptions in Theorem~\ref{cor: cor 1} and Lemma~\ref{cor: 1} are satisfied. If $\E\{\lvert f^*(\bX) + \delta\pi^*(\bX)^{-1}\{Y - f^*(\bX)\}\rvert^2\}<\infty$, then
	$
	\sqrt{n}(\widehat{\theta}_{AIPW}-\theta^*)\to \rm{N}(0,\sigma^2),
	$
	in distribution, where $\pi^*(\bX) = \Pr(\delta=1\mid \bX)$ and $\sigma^2 = \Var{[f^*(\bX) + \delta\pi^*(\bX)^{-1}\{Y-f^*(\bX)\}]}$. 
\end{theorem}
It is worthy pointing out that the derived result is particularly attractive given the fact that the central limit theorem is built by nonparametric estimation of $f^*(\bx)$ with  diverging dimension, and to our knowledge, such a result is novel in literature. More importantly, the variance  term $\sigma^2$ can be estimated by the sample variance of $\{\widehat{f}_0(\bx_i) + \delta_i\widehat{\pi}(\bx_i)^{-1}\{y_i-\widehat{f}_0(\bx_i)\}:i=1,\ldots,n\}$:
\begin{equation*}
	\widehat{\sigma}_{AIPW}^2 = \frac{1}{n-1}\sum_{i=1}^n(\widehat{y}_i-\widehat{\theta}_{AIPW})^2,\label{eq: 3.1}
\end{equation*}
where $\widehat{y}_i = \widehat{f}_0(\bx_i) + \delta_i\widehat{\pi}(\bx_i)^{-1}\{y_i-\widehat{f}_0(\bx_i)\}$. Thus, based on Theorem~\ref{thm:clt} and the estimated variance $\widehat{\sigma}_{AIPW}^2$, we can also obtain the interval estimators of $\theta^*$.

\section{Numerical analysis	}
\label{others}

In this section, we compared the numerical performance of the proposed AIPW estimator, denoted as Prop, against several state-of-the-art competitors under two simulated experiments and a real-data application. For Prop,	in all the scenarios,  we applied a Gaussian kernel, $K(\bx,\bu) = \exp{\left(-{\|\bx-\bu\|_2^2}/(2\sigma_n^2)\right)}$ with $\sigma_n$ being the median of all the pairwise distances among the covariates \citep{Jaakkola1999}. As suggested by \citet{He20192}, we also applied the  stability-based selection criterion \citep{SunWW2013} to determine the thresholding value $v_n$ and set the ridge parameter $\lambda_n=0.001$ for the employed sparse learning algorithm.

\subsection{Simulated experiments}\label{sec: synthetic}
In this section, we considered  $n\in\{800,1\,000\}$ and $p\in\{400,2\,000\}$, and  covariates were generated by $x_{il}\sim\mbox{U}(-0.5,0.5)$ for $i=1,\ldots,n$ and $l=1,\ldots,p$, where  $x_{il}$ denoted the $l$th element of $\bx_i$, and $\mbox{U}(-0.5,0.5)$ denoted a uniform distribution over $[-0.5,0.5]$.  The following regression models were applied to generate the response of interest:
\begin{enumerate}
	\item[M1.] Linear regression model: $y_i = 5x_{i1} + 6x_{i2}+ 4x_{i3}+ 4x_{i4} + \epsilon_i$ with $\epsilon_i\sim \rm{N}(0,1)$. 
	\item[M2.] Nonlinear regression model: $y_i =6x_{i1} + 4(2x_{i2}+1)(2x_{i3}-1)  + 6h(x_{i4})+ 5\sin(x_{i5}\pi)/\{2-\sin(x_{i5}\pi)\} +  \epsilon_i$, where $h(x) = 0.1\sin(x_{i4}\pi) + 0.2\cos(x_{i4}\pi) + 0.3\sin(x_{i4}\pi)^2 + 0.4\cos(x_{i4}\pi)^3+ 0.5\sin(x_{i4}\pi)^3$, and $\epsilon_i\sim \rm{N}(0,1)$. 
\end{enumerate}
For $i=1,\ldots,n$, the response indicator $\delta_i$ was generated by a Bernoulli distribution with success probability $\pi^*(\bx_i)$, which was obtained by the following models:
\begin{enumerate}
	\item[R1.] Logistic response model: $\mbox{logit}\{\pi^*(\bx_i)\} = -0.1 + 2x_{i1} + 2x_{i3}$.
	\item[R2.] Multi-modal response model  $\pi^*(\bx_i) = \sin(6x_{i2} + 8x_{i4})/3+0.5$.
\end{enumerate}
The linear regression model M1 is commonly assumed in practice \citep{Fan2001}. The nonlinear regression model M2, however, is more complex, and the interaction effect is also taken into consideration. The logistic response model R1 is  widely used in practice. However, the response model R2 violates (\ref{eq: response prob}), so it is used to test the robustness of the proposed AIPW estimator.

The primary interest was to estimate $\theta^* = \E(Y)$.  For the regression model M1, we had $\theta^*=0$. However, instead of deriving $\theta^*$ analytically, we used $\tilde{\theta} = L^{-1}\sum_{l=1}^Ly_{l}$ as the ``true value'' for the regression model M2, where $\{y_l:l=1,\ldots,L\}$ was a random sample of size $L=1\,000\,000$. The following competitors were considered:
\begin{enumerate}
	\item[CC.] The sample mean of the complete cases, $\widehat{\theta}_{cc}=m^{-1}\sum_{i=1}^n\delta_iy_i$, where $m = \sum_{i=1}^n\delta_i$.
	\item[PS.] Conventional propensity score  estimator
	$
	\widehat{\theta}_{ps} = {n}^{-1}\sum_{i=1}^n\delta_i\pi^{-1}(\bx_i;\widehat{\bbeta})y_i,
	$
	where  $\widehat{\bbeta}^\top = (\widehat{\beta}_0,\widehat{\bbeta}_1^\top)$ solves 
	$
	\sum_{i=1}^n\{\delta_i-\pi(\bx_i;\widehat{\bbeta})\}(1,\bx_i^\top)=0
	$ without consideration of the sparsity.
	\item[DI.] Deterministic imputation using kernel ridge regression \citep{wangkim2020}
	$
	\widehat{\theta}_{di} = {n}^{-1}\sum_{i=1}^n\{\delta_iy_i + (1-\delta_i)\widehat{f}(\bx_i)\},
	$
	where $\widehat{f}(\bx)$ is the fitted kernel ridge regression model based on the fully observed data $\{(\bx_i,y_i):\delta_i=1\}$ without employing sparse learning.
	\item[NAIPW.] Naive AIPW estimator
	$
	\widehat{\theta}_{AIPW1}={n}^{-1}\sum_{i=1}^n\{\widehat{f}(\bx_i) + \delta_i\pi^{-1}(\bx_i;\widehat{\bbeta})\{y_i - \widehat{f}(\bx_i)\}\},\label{eq: DR1}
	$
	where $\widehat{\bbeta}$ is the same as that in the  PS estimator without consideration of the sparsity, and $\widehat{f}(\bx)$ is the same as that in the DI estimator without employing sparse learning.
\end{enumerate}
The CC estimator completely ignore the unobserved data, leading to a biased estimator if $\E (\delta_i\mid \bx_i)$ involves covariates used in the regression model. PS estimator  is widely used in causal inference \citep{rosenbaum1983central} and missing data analysis \citep{WOOLDRIDGE20071281}. The imputation methods are commonly used to provide a complete dataset, especially in survey sampling; see \citet[Chapter~4]{kim2013statistical} for details. Recently, \citet{wangkim2020} has proposed a kernel-based deterministic imputation method, and we consider their method for comparison as well. Except for the proposed AIPW estimator, we also considered the naive AIPW estimator based on the conventional propensity score estimator and the deterministic imputation estimator.  

We conducted $M=500$ Monte Carlo simulations for each estimator under different model setups. First, we compared different estimators in terms of  the Monte Carlo bias and the Monte Carlo standard error: 
\begin{align*}
	\mbox{Bias} &= \bar{\theta}_n^{(M)} - \theta,\\ \mbox{SE}&=\left\{\frac{1}{M-1}\sum_{m=1}^M(\widehat{\theta}_n^{(m)}-\bar{\theta}_n^{(M)})^2\right\}^{1/2},
\end{align*}
where $\bar{\theta}_n^{(M)}=M^{-1}\sum_{m=1}^M\widehat{\theta}_n^{(m)}$, $\widehat{\theta}_n^{(m)}$ was a specific estimator of $\theta$ for the $m$th Monte Carlo simulation. Simulation results were summarized in Table~\ref{tab: BIAS and SE}. The CC estimator is biased since the response of interest $y_i$ is correlated with the response index $\delta_i$. Even though a logistic model is correctly specified for the response model R1, the PS estimator is still biased or even unrealistic due to the curse of dimensionality. Since the NAIPW estimator was obtained using the same response model as the PS estimator, it is also questionable, especially when the sample size is small and the number of useless covariates is large. Since the response probability was not used by the DI estimator, it does not suffer the same problem as the PS estimator. However, even under the linear regression model, the bias of the DI estimator may not be negligible compared with its standard error.  Compared with its competitors, the proposed AIPW estimator performs the best since it has the smallest bias under most model setups, and its standard error is reasonably small.

\begin{table}[h]
	\centering
	\caption{Summary of the Monte Carlo bias (Bias) and standard error  (SE) corresponding to the five estimators under different model setups, and the unit is 0.1.  For ``Model'',  ``C1--C4'' represent (M1, R1),(M2, R1), (M2, R2) and (M2, R2), respectively. For ``Size'', ``I--IV'' corresponds to $(n,p) = (800,400)$, $(1\,000,400)$, $(800,2\,000)$, and $(1\,000,2\,000)$, respectively. Notation ``-'' is used when the absolute value of either bias or standard error is greater than $100$. }\label{tab: BIAS and SE}
	\begin{tabular}{@{}c@{\hskip 0.1in}c@{\hskip 0.1in}r@{\hskip 0.1in}r@{\hskip 0.1in}r@{\hskip 0.1in}r@{\hskip 0.1in}r@{}}
		\hline
		Model&Size&CC&PS&DI&NAIPW&Prop\\
		\hline
		& I & 6.9 (1.5)& - (-)& 0.9 (1.1)& - (-)& { 0.3} (1.2) \\ 
		& II & 7.0 (1.3)& - (-)& 0.9 (1.0)& - (-)& 0.1 (1.0) \\ 
		& III & 7.0 (1.5)& 3.4 (0.7)& 1.1 (1.2)& 1.1 (1.2)& 0.4 (1.3) \\ 
		\multirow{-4}{*}{C1}& IV & 6.9 (1.3)& 3.3 (0.6)& 1.0 (1.0)& 1.0 (1.0)& 0.2 (1.1)\\ \\
		& I & -0.7 (2.8)& - (-)&  0.1 (2.0)& - (-)&  0.0 (2.0)\\ 
		& II & -0.8 (2.4)& - (-)& -0.1 (1.7)& - (-)& -0.1 (1.6)\\ 
		& III & -0.7 (2.6)&  2.5 (1.3)&  0.3 (2.1)&  0.3 (2.1)&  0.0 (2.0)\\ 
		\multirow{-4}{*}{C2}& IV & -0.6 (2.3)&  2.6 (1.1)&  0.0 (1.7)&  0.0 (1.7) & -0.1 (1.7)\\ \\
		& I & -1.2 (1.5)& - (-)& -0.3 (1.1)& - (-)& -0.2 (1.1)\\ 
		& II & -1.3 (1.3)&  0.3 (57.1)& -0.3 (1.0)&  0.2 (8.4)& -0.2 (1.0) \\ 
		& III & -1.2 (1.5)& -0.6 (0.7)& -0.2 (1.1)& -0.2 (1.1)& -0.1 (1.1) \\ 
		\multirow{-4}{*}{C3} & IV & -1.2 (1.3)& -0.6 (0.7)& -0.2 (1.0)& -0.2 (1.0) & -0.2 (1.0)  \\ \\
		& I & 1.7 (2.6)& - (-)& 0.2 (1.9)& - (-)& 0.0 (1.9)  \\ 
		& II &  1.7 (83.6)& -3.0 (2.2)&  0.2 (1.6)& -0.2 (17.9)&  0.0 (1.7)  \\ 
		& III & 1.8 (2.6)& 3.6 (1.3)& 0.3 (2.0)& 0.3 (2.0)& 0.1 (2.0) \\ 
		\multirow{-4}{*}{C4}& IV & 1.9 (2.2)& 3.6 (1.1)& 0.3 (1.6)& 0.3 (1.6)& 0.0 (1.6) \\ 
		\hline
	\end{tabular}
\end{table}

Next, the proposed AIPW  estimator was evaluated   by the relative bias of the variance estimator and its coverage rate of a 95\% confidence interval:
\begin{align*}
	\mbox{RB} &= \frac{\bar{\sigma}_n^{2(M)} - \mbox{SE}^2}{\mbox{SE}^2},\\
	\mbox{CR} &= \frac{1}{M}\sum_{m=1}^M \mathbb{I}(\widehat{\theta}_{n}^{(m)} - 1.96\widehat{\sigma}_n^{(m)}\leq \theta\leq \widehat{\theta}_{n}^{(m)} + 1.96\widehat{\sigma}_n^{(m)}),
\end{align*}
where  $\bar{\sigma}_n^{2(M)} = M^{-1}\sum_{i=1}^M\widehat{\sigma}_n^{2(m)}$,  $\widehat{\sigma}_n^{2(m)}$ is the   variance estimator for the $m$th Monte Carlo simulation, $\widehat{\sigma}_n^{(m)}$ is the square root of  $\widehat{\sigma}_n^{2(m)}$, and $\mathbb{I}(a\leq x\leq b)$ is an indicator function of $u$ for given $a\leq b$, and it takes value 1 if $u\in[a,b]$ and 0 otherwise. We conducted  500 Monte Carlo simulations, and  Table~\ref{tab: coverage rate} summarized the corresponding results. Generally, as the sample size increases from $n=800$ to $n=1\,000$, the relative bias of the variance estimator decreases, and it is  negligible if sample size $n=1\,000$. Thus, the proposed variance estimator performs well, especially when the sample size is large. The coverage rates are close to their nominal truth 0.95 when sample size is large. Since the variance of the proposed AIPW estimator  is under-estimated for  the setup with regression model $M1$ and response model $R1$, the corresponding coverage rate is  much lower than 0.95.  As the sample size increases, however,  the coverage rate gets closer to its nominal truth. For the two setups with response model R2,  a logistic regression model is  wrongly specified for the response indicator. However, the absolute values of the relative bias of the variance estimator are generally less than 0.05 and the corresponding coverage rates are close to the nominal truth 0.95, specially when the sample size is large.   Thus, the proposed AIPW estimator is indeed robust against a wrongly specified response model.   

\begin{table}[h]
	\centering
	
	\caption{Relative bias (RB) of the variance estimator and coverage rate (CR) of the 95\% confidence interval of $\theta$ for the proposed   method.  For ``Model'',  ``C1--C4'' represent (M1, R1),(M2, R1), (M2, R2) and (M2, R2), respectively.  ``I--IV'' corresponds to $(n,p) = (800,400)$, $(1\,000,400)$, $(800,2\,000)$, and $(1\,000,2\,000)$, respectively. }\label{tab: coverage rate}
	
	\begin{tabular}{@{}c@{\hskip 0.1in}rrrrrrrrr@{}}
		\hline
		& \multicolumn{4}{c}{RB} && \multicolumn{4}{c}{CR}\\ 
		\multirow{-2}{*}{Model}&I&II&III&IV&&I&II&III&IV\\
		\hline
		C1 & -0.10 & -0.07 & -0.33 & -0.20 & &  0.92 & 0.94 & 0.88 & 0.92 \\ 
		C2 & -0.13 & 0.01 & -0.14 & -0.01 & &  0.94 & 0.94 & 0.93 & 0.95 \\ 
		C3 & 0.04 & -0.02 & -0.03 & 0.05 & &  0.95 & 0.95 & 0.95 & 0.96 \\ 
		C4 & -0.10 & -0.01 & -0.14 & 0.05 & &  0.94 & 0.94 & 0.93 & 0.96 \\ 
		\hline
	\end{tabular}
\end{table}

\subsection{Application to a supermarket dataset}

In this section, the proposed AIPW estimator  and its competitors were applied to a supermarket dataset \citep{Wang2009}, which was collected from a major supermarket located in northern China, consisting of daily sale records of $p=6\,398$ products on $n=464$ days. This data included almost all kinds of daily necessities and the response of interest was the number of customers on each day, and the covariates are the daily sale volumes of each product. For simplicity, denote $y_i$ and $\bx_i=(x_{i1},\ldots,x_{ip})^{\top}$ to the be the response of interest and the corresponding covariate  for the $i$th day.  In this section, we were interested in estimating the average number of customers visiting the supermarket. This dataset was fully observed, and it was  studentized before analyzing. Thus, the sample mean $\widehat{\theta} = n^{-1}\sum_{i=1}^ny_i=0$  served as a benchmark.

To compare the performance of the proposed AIPW estimator with other competitors, we considered  the following missing mechanism for $y_i$:
\begin{eqnarray}
	\mbox{logit}\{\pi(\bx_i)\} = 1 - 0.6x_{i5} - x_{i6} + 0.5x_{i10},\label{eq: app 01}
\end{eqnarray}
and $y_i$ was treated as  observed if and only if $\delta_i=1$, where $\delta_i=1$ with probability $\pi_i$.  The  mechanism in (\ref{eq: app 01}) was MAR, and the corresponding covariates were identified as informative for estimating the response of interest by \citet{He20192}. Then, instead of observing the whole data, we assumed that all the covariates and only  $\{y_i:\delta_i=1\}$ were available, and the resulting response rate was about 0.70.  

We generated 500 incomplete datasets using (\ref{eq: app 01}) and compare the estimators in Section~\ref{sec: synthetic}.  Table~\ref{tab: application} summarized the average of the estimators and the corresponding standard error. 
The CC estimator is highly biased since it ignores the missing mechanism. The performance of the PS estimator is also questionable in that the response model using all covariates results in overfitting. The DI, NAIPW and  Prop  estimators outperform the CC and PS estimators since their estimates are much closer to 0. However, the standard error of the Prop estimator is much smaller than the other two, illustrating the superior of the proposed AIPW estimator.

\begin{table}[ht]
	\centering
	\caption{The average and standard error of 500 incomplete  datasets for estimating the  number of customers visiting the supermarket. Since studentized is applied, the sample mean 0 serves as the benchmark. }\label{tab: application}
	\begin{tabular}{@{}crrrrr@{}}
		\hline
		&CC&PS&DI&NAIPW&Prop\\
		\hline
		Estimate& -0.20&-0.13& -0.04 & -0.04 & -0.04\\
		Standard error &0.03 & 0.02& 0.03&0.03&0.01\\
		\hline
	\end{tabular}
\end{table}

	\section{Conclusion}
\label{Conclude}
In this paper, we propose a novel AIPW estimator to infer the population mean, which incorporates an efficient nonparametric imputation with  sparse structure and a penalized propensity score estimator under the assumption of  missing at random. The proposed method is computationally efficient and allows the dimension diverging. More importantly, the estimation consistency as well as the corresponding central limit theorem are established under  regularity assumptions.  Its superior is also supported by several simulated examples and one application to a supermarket dataset.

\section*{Acknowledgement}
Xin He's research is supported by NSFC Grant No. 11901375; Xiaojun Mao's research is supported by NSFC Grant No. 12001109 and 92046021, the Science and Technology Commission of Shanghai Municipality grant 20dz1200600. Zhonglei Wang's research is supported by NSFC Grant No. 11901487 and 72033002, Fundamental Scientific Center of National Natural Science Foundation of China Grant No. 71988101.
\appendix

\section{Block co-ordinate gradient descent algorithm}	\label{append: BCGD ALG}

For $g$th group of $\bbeta_1$, consider a vector $\bd$ such that $\bd_k=0$ for $k\neq g$, and assume that the $\mbox{df}_g\times \mbox{df}_g$ submatrix is of the form $H_{gg}^{(t)} = h_g^{(t)}I_{\scriptsize\mbox{df}_g}$ for some scalar $h_g^{(t)}$, where $I_m$ is an $m\times m$ identity matrix.

If $\lVert \nabla l(\widehat{\bbeta}^{(t)})_g - h_g^{(t)}\widehat{\bbeta}^{(t)}_g\rVert_2\leq \lambda_2 \mbox{df}_g^{1/2}$, let $\bd_g^{(t)} = -\widehat{\bbeta}^{(t)}_g$. Otherwise, 
$$
\bd_g^{(t)} = -\frac{1}{h_g^{(t)}} \left\{  \nabla l(\widehat{\bbeta}^{(t)})_g  - \lambda_2 \mbox{df}_g^{1/2} \frac{  \nabla l(\widehat{\bbeta}^{(t)})_g  - h_g^{(t)}\widehat{\bbeta}^{(t)}_g}{\lVert \nabla l(\widehat{\bbeta}^{(t)})_g - h_g^{(t)}\widehat{\bbeta}^{(t)}_g\rVert_2}\right\}.
$$
If $\bd^{(t)}\neq 0$,  let $\bbeta^{(t+1)} = \bbeta^{(t)}  + \alpha^{(t)}\bd^{(t)} $, where $\alpha^{(t)}$ is the largest value among $\{\alpha_0\delta^l:l\geq0\}$ such that 
$$
S_{\lambda_2}(\bbeta^{(t)} + \alpha^{(t)}\bd^{(t)}) - S_{\lambda_2}(\bbeta^{(t)})\leq \alpha^{(t)}\sigma\Delta^{(t)},
$$
$\delta\in(0,1)$, $\sigma\in(0,1)$, $\alpha_0>0$, and 
$$
\Delta^{(t)}  = -\left(\bd^{(t)}\right)^{\top} \nabla l(\widehat{\bbeta}^{(t)}) + \lambda_2\mbox{df}_g^{1/2}\lVert \widehat{\bbeta}^{(t)}_g + \bd^{(t)}_g\rVert_2 - \lambda_2\mbox{df}_g^{1/2}\lVert \widehat{\bbeta}^{(t)}_g\rVert_2.
$$
See \citet{Meier2008}  for details.

\section{Proofs}
	\label{sec: proof app}

\begin{prop}\label{prop1}
	Suppose  Assumptions 1--2 are satisfied. Then, with probability at least $1-\delta_n/2$, there holds
	\begin{eqnarray}\label{thm12}
		\big \|\widehat{f}-f^* \big \|_{K} \leq 2\log\left(\frac{8}{\delta_n}\right) \Big [ 	\frac{3\kappa_{1,p}}{n^{1/2}\lambda_n} \left\{\kappa_{1,p}\|f^*\|_K+q^{-1}\left(\log\frac{4c_1 n}{\delta_n}\right)\right\} \notag\\
		+ \lambda_n^{r-{1}/{2}} \|L^{-r}_{K}f^*\|_2 \Big ].\notag
	\end{eqnarray}
\end{prop}
The proof of Proposition  \ref{prop1} is similar as that in \citet{He20192} and thus we omit it here.

\begin{proof}[Proof of Lemma \ref{thm1}]
	The proof of Lemma \ref{thm1} is similar as that in \citet{He20192} by using Proposition  \ref{prop1}, the property of  Hilbert-Schmidt operators and the concentration inequalities in Hilbert-Schmidt operator space. Thus we omit the detail here. 
\end{proof}

\begin{proof}[Proof of Lemma \ref{thm3}]
	The proof of Lemma \ref{thm3} is similar as that in \citet{He20192},  and thus we omit the detail here. 
\end{proof}

\begin{proof}[Proof of Theorem \ref{cor: cor 1}]
	Define the event that 
	\begin{eqnarray}
	{\cal C}_1=\left \{  \|\widehat{f}_0 -f^* \|_{\infty} > c_{m, 2} \max \left \{ \kappa_{1,p_0}\|f^*\|_K, q^{-1}\left(\log\frac{2c_1 m}{\delta_m}\right) \right \}\right. \notag \\ 
	\times\left.{\log\left(  \frac{4}{\delta_m}\right)} { m^{-\frac{2r-1}{2(2r+1)}}  } \right \}.
	\end{eqnarray}
	
	Then, the probability $\Pr({\cal C}_1)$ can be decomposed as 
	\begin{align*}
		\Pr \left ({\cal C}_1 \right )&= \Pr\left ({\cal C}_1, \big\{\widehat{\cal A}_{v_m}={\cal A}^* \big \} \right ) + \Pr\left ({\cal C}_1, \big\{\widehat{\cal A}_{v_m}={\cal A}^* \big \} \right ) \\
		&= \Pr\left ({\cal C}_1\mid \big\{\widehat{\cal A}_{v_m}={\cal A}^* \big \} \right ) \Pr\left (\widehat{\cal A}_{v_m}={\cal A}^*\right ) \notag \\
		&\quad+ \Pr\left ({\cal C}_1\mid \big\{\widehat{\cal A}_{v_m}={\cal A}^* \big \} \right ) \Pr\left ( \widehat{\cal A}_{v_m}\neq{\cal A}^* \right ) \\
		&\leq \Pr\left ({\cal C}_1\mid \big\{\widehat{\cal A}_{v_m}={\cal A}^* \big \} \right ) (1-\Delta_m) +  \Delta_m.
	\end{align*}
	By Lemma \ref{thm3}, we have $\Delta_m \rightarrow 0$ and $(1-\Delta_m) \rightarrow 1$.
	For $ \Pr({\cal C}_1\mid \big\{\widehat{\cal A}_{v_m}={\cal A}^* \big \} )$, by applying the proof in Proposition \ref{thm1} conditioning on $\big\{\widehat{\cal A}_{v_m}={\cal A}^* \big\}$, with probability at least $1-\delta_m$, there holds
	\begin{eqnarray}
	\big \|\widehat{f}-f^* \big \|_{K} &\leq  c_{m, 2} \max \left \{ \kappa_{1,p_0}\|f^*\|_K, q^{-1}(\log\frac{2c_1 m}{\delta_m}) \right \} \notag \\ 
	&\times{\log\left(  \frac{4}{\delta_m}\right)} { m^{-\frac{2r-1}{2(2r+1)}}  },\notag
	\end{eqnarray}
	which implies $\Pr({\cal C}_1\mid \big\{\widehat{\cal A}_{v_m}={\cal A}^* \big \} )\leq \delta_m$.
	Combining the above results, we have $\Pr({\cal C}_1 )\leq \delta_m+\Delta_m$. This completes the proof of the first part in Theorem \ref{cor: cor 1}. 
	
	Additionally, by Assumption 4, we have $m=O(n)$, and if we take $r=1$ and assume that $p_0=O(1)$ and $\epsilon$ has sub-Gaussian or sub-exponential tail, there holds 
	$$
	\big \|\widehat{f}-f^* \big \|_{K} =O_{\mathbb{P}}(  n^{-\frac{1}{6}}\log n  ).
	$$
	Note that 
	\begin{eqnarray}
	\big \|\widehat{f}-f^* \big \|_{\infty}&=&\sup_{\bx} \lvert \widehat{f}(\bx)-f^*(\bx)\rvert\notag \\
	&=&\sup_{\bx} \lvert\langle \widehat{f}-f^* , K_{\bx}\rangle_K\rvert \leq  \kappa_{1,p_0} \big \|\widehat{f}-f^* \big \|_{K},\notag
	\end{eqnarray}
	which completes the proof.
\end{proof}

\begin{proof}[Proof of Lemma~\ref{cor: 1}]
	By Assumptions 4--7, \citet{Meier2008} showed (\ref{eq: append 1}). Denote $g(x)=\{1+\exp(-x)\}^{-1}$, and we can show that 
	$
	g'(x)={\mbox{d}g(x)}/{\mbox{d}x} = g(x)\{1-g(x)\}. 
	$
	That is, 
	\begin{equation}\label{eq: append 2}
		\lvert g'(x)\rvert\leq 1,
	\end{equation}
	for any $x$ by the fact that $0\leq g(x)\leq 1$. Thus, by (\ref{eq: append 2}) and the mean value theorem, we conclude that $g(x)$ is Lipschitz continuous in the sense that 
	\begin{equation*}\label{eq: append 3}
		\lvert g(x_1) - g(x_2)\rvert\leq \lvert x_1-x_2\rvert,
	\end{equation*}
	for any $x_1$ and $x_2$ in $\mathbb{R}$.
	By noting the fact that $\widehat{\pi}(\bx) = g\{\eta_{\widehat{\bbeta}_{\lambda_2}}(\bx)\}$ and $\pi^*(\bx) = g\{\eta_{\bbeta_0}(\bx)\}$, by (\ref{eq: append 1}), we have 
	\begin{eqnarray}\label{eq: append 4}
		\E\{ \lvert\widehat{\pi}(\bX) - \pi^*(\bX)\rvert^2\}&\leq& \E\{\lvert\eta_{\widehat{\bbeta}_{\lambda_2}}(\bX) - \eta_{{\bbeta}_0}(\bX) \rvert^2\}\notag \\
		&=& O_{\mathbb{P}}(n^{-2/3-2\zeta}),
	\end{eqnarray}
	where the expectation is taken conditional on $\widehat{\bbeta}_{\lambda_2}$.
	By (\ref{eq: append 4}), we have shown Lemma~\ref{cor: 1}.

\end{proof}

\begin{lemma}\label{lemma: A1 beta dif}
	Suppose  Assumptions 4--7 are satisfied. Then, given $\widehat{\bbeta}_{\lambda_2}$, there holds
	$$
	\max\{\lvert\widehat{\beta}_{k}-\beta_{k}^*\rvert:k=0,\ldots,p\} = O_{\mathbb{P}}(n^{-1/3-\zeta}),
	$$
	where $\widehat{\beta}_{k}$ and $\beta_{k}^*$ are the $(k+1)$th component of $\widehat{\bbeta}_{\lambda_2}$ and $\bbeta^*$, respectively.
\end{lemma}
\begin{proof}[Proof of Lemma~\ref{lemma: A1 beta dif}]
	Given the estimated parameters for the response model, (\ref{eq: append 1}) can be re-expressed as 
	\begin{eqnarray}
		&&\E\{\lvert\eta_{\widehat{\bbeta}_{\lambda_2}}(\bX) - \eta_{{\bbeta}^*}(\bX) \rvert^2\} \notag \\
		&=& \E\{(\widehat{\bbeta}_{\lambda_2}-\bbeta^*)^{\top}\bX\bX^{\top}(\widehat{\bbeta}_{\lambda_2}-\bbeta^*)\}\notag \\ 
		&=& (\widehat{\bbeta}_{\lambda_2}-\bbeta^*)^{\top} \E(\bX\bX^{\top}) (\widehat{\bbeta}_{\lambda_2}-\bbeta^*)\}\notag \\ 
		&=& O_{\mathbb{P}}(n^{-2/3-2\zeta}),\label{eq: Lemma A1 1}
	\end{eqnarray}
	By Assumption~A5 and (\ref{eq: Lemma A1 1}), we conclude that 
	\begin{eqnarray}
		(\widehat{\bbeta}_{\lambda_2}-\bbeta^*)^{\top}  (\widehat{\bbeta}_{\lambda_2}-\bbeta^*) &=&(\widehat{\beta}_0 - \beta_0^*)^2 + \sum_{k=1}^p(\widehat{\beta}_{k}-\beta_{k}^*)^2\notag \\ 
		&=&  O_{\mathbb{P}}(n^{-2/3-2\zeta}).\label{eq: Lemma A1 2}
	\end{eqnarray}
	Notice that 
	\begin{equation}
		\max\{(\widehat{\beta}_{k}-\beta_{k}^*)^2:k=0,\ldots,p\}\leq  (\widehat{\beta}_0 - \beta_0^*)^2 + \sum_{k=1}^p(\widehat{\beta}_{k}-\beta_{k}^*)^2.\label{eq: Lemma A1 3}
	\end{equation}
	Thus, we have proved Lemma~\ref{lemma: A1 beta dif} by (\ref{eq: Lemma A1 2}) and (\ref{eq: Lemma A1 3}).
\end{proof}

\begin{proof}[Proof of Theorem \ref{thm:clt}]
	
	For simplicity, denote $\pi^*_i = \pi^*(\bx_i)$ and $\widehat{\pi}_i = \widehat{\pi}(\bx_i) = (1+ \exp[-\{(1,\bx_i^{\top})\widehat{\bbeta}_{\lambda_2}\}])^{-1}$.
	What if we consider 
	$$
	\widehat{\theta}_{AIPW} = \frac{1}{n}\sum_{i=1}^n\left[\hmi + \frac{\delta_i}{\widehat{\pi}_i}\left\{y_i - \hmi\right\}\right],
	$$
	where $\hpii$ is an estimator of $\Pr(\delta_i=1\mid \bx_i)$ by the group lasso for logistic regression.
	
	Then, we have 
	\begin{eqnarray}
		\widehat{\theta}_{AIPW} &=& \frac{1}{n}\sum_{i=1}^n\left[\mi + \left\{\hmi - \mi\right\} + \frac{\delta_i}{\pii}\{y_i - \mi\} \right.\notag \\ 
		&&+ \frac{\delta_i}{\hpii}\{y_i - \mi\} - \frac{\delta_i}{\pii}\{y_i - \mi\} \notag \\
		&&\left.+\frac{\delta_i}{\hpii}\{\mi-\hmi\}\right]\notag \\ 
		&=& \frac{1}{n}\sum_{i=1}^n\left[\mi + \frac{\delta_i}{\pii}\{y_i - \mi\} \right]\notag \\ 
		&&+ \frac{1}{n}\sum_{i=1}^n\left\{1-\frac{\delta_i}{\hpii}\right\}\left\{\hmi - \mi\right\} \notag \\ 
		&& + \frac{1}{n}\sum_{i=1}^n\left[ \frac{\delta_i}{\hpii}- \frac{\delta_i}{\pii}\right]\{y_i - \mi\}\notag \\ 
		&=&\frac{1}{n}\sum_{i=1}^n\left[\mi + \frac{\delta_i}{\pii}\{y_i - \mi\} \right]\notag \\  
		&& + \frac{1}{n}\sum_{i=1}^n\left\{1-\frac{\delta_i}{\hpii}\right\}\left\{\hmi - \mi\right\} \notag \\ 
		&& + \frac{1}{n}\sum_{i=1}^n\left[ \frac{\delta_i}{\hpii}- \frac{\delta_i}{\pii}\right]\epsilon_i,\label{eq: A1}
	\end{eqnarray}
	where $\epsilon_i = y_i - \mi$. 
	
	First, we consider the first term of (\ref{eq: A1}), and we have 
	\begin{eqnarray}
		&&\E\left[f^*(\bX) + \frac{\delta}{\pi^*(\bX)}\{Y - f^*(\bX)\} \right]\notag \\ 
		&=& \E\left(\E\left[f^*(\bX) + \frac{\delta}{\pi^*(\bX)}\{Y - f^*(\bX)\} \right] \mid \bX,Y\right)\notag \\ 
		&=& \E(Y),\notag
	\end{eqnarray}
	where $\delta$ is a binary random variable with success probability $\pi^*(\bX)$ conditional on $\bX$.
	Since  $\E\{\lvert f^*(\bX) + \delta\pi^*(\bX)^{-1}\{Y - f^*(\bX)\}\rvert^2\}<\infty$, by the classical central limit theorem \citep[Example~2.1]{van2000asymptotic}, we have  
	\begin{equation}
		\sqrt{n}\left(\frac{1}{n}\sum_{i=1}^n\left[\mi + \frac{\delta_i}{\pii}\{y_i - \mi\} \right] -\theta^*\right)\to \rm{N}(0,\sigma^2),\label{eq: CLT first}
	\end{equation}
	in distribution under regularity conditions, where $\theta = \E(y)$ and $\sigma^2$ is  to be estimated. 
	
	Next, we consider the third term of (\ref{eq: A1}) .
	\begin{eqnarray}
		\frac{1}{n}\sum_{i=1}^n\left[ \frac{\delta_i}{\hpii}- \frac{\delta_i}{\pii}\right]\epsilon_i &=&  \frac{1}{n}\sum_{i=1}^n\frac{\delta_i\epsilon_i}{\pii}\frac{\pii-\hpii}{\hpii}. \notag\label{eq: A2}
	\end{eqnarray}
	By Assumption 4 and Lemma~\ref{cor: 1}, we conclude that $(\pii-\hpii)\hpii^{-1}=o_{\mathbb{P}}(1)$ uniformly for $i=1,\ldots,n$. Since 
	$$ \frac{1}{n}\sum_{i=1}^n\frac{\delta_i\epsilon_i}{\pii}=O_{\mathbb{P}}(n^{-1/2}),$$
	we conclude that 
	\begin{equation}
		\frac{1}{n}\sum_{i=1}^n\left[ \frac{\delta_i}{\hpii}- \frac{\delta_i}{\pii}\right]\epsilon_i=o_{\mathbb{P}}(n^{-1/2}).\label{eq: clt third}
	\end{equation}
	
	Now, we focus on the second term  of (\ref{eq: A1}), and consider 
	\begin{eqnarray}
		&&\frac{1}{n}\sum_{i=1}^n\left\{1-\frac{\delta_i}{\hpii}\right\}\left\{\hmi - \mi\right\}\notag \\
		&=& \frac{1}{n}\sum_{i=1}^n\left\{1-\frac{\delta_i}{\pii}\right\}\left\{\hmi - \mi\right\} \notag\\
		&&+  \frac{1}{n}\sum_{i=1}^n\frac{\delta_i}{{\pii}\hpii}(\hpii - \pii)\left\{\hmi - \mi\right\}.\notag \\\label{eq: A3}
	\end{eqnarray}
	By Theorem~\ref{cor: cor 1}, $\hmi - \mi = O_{\mathbb{P}}(\log(n)n^{-1/6})$ uniformly for $i=1,\ldots,n$. Since $n^{-1}\sum_{i=1}^n\{1-\delta_i(\pi_i^*)^{-1}\}=O_{\mathbb{P}}(n^{-1/2})$, the first term of  is of (\ref{eq: A3}) is of the order $o_{\mathbb{P}}(n^{-1/2})$. Besides, to show the second term of (\ref{eq: A3}) is also of the order $o_{\mathbb{P}}(n^{-1/2})$, it is enough to show  
	\begin{equation}
		\frac{1}{n}\sum_{i=1}^n\frac{\delta_i}{{\pii}\hpii}(\hpii - \pii) = O_{\mathbb{P}}(n^{-1/3-\zeta/2}),\label{eq: A4}
	\end{equation} 
	where $\zeta$ is in Assumption~7.
	
	By  Lemma~\ref{lemma: A1 beta dif}, we conclude that 
	$
	\max\{\lvert\widehat{\beta}_{k}-\beta_{k}^*\rvert:k=0,\ldots,p\}  = o_p(1).
	$
	Denote $A_n$ to be the event that $ \{	\max\{\lvert\widehat{\beta}_{k}-\beta_{k}^*\rvert:k=0,\ldots,p\}\geq C_\kappa \}$, where $C_\kappa$ is a positive constant such that $\min\{\widehat{\bbeta}_{\lambda_2}(\bx) \geq \kappa/2:\bx\in\mathcal{X}\}$. The existence of $C_\kappa$ is guaranteed by the compactness of $\mathcal{X}$ and Assumption~A4.  Then,  we have $P(A_n)\to0$ as $n\to \infty$. On $\bx\in A_n^C$, we conclude $\hat{\pi}(\bx)\geq \kappa/2$.
	
	Since $\{\bx_1,\ldots,\bx_n\}$ is a random sample, given $\widehat{\bbeta}_{\lambda_2}$, for any positive constant $C$, we consider 
	\begin{eqnarray}
		&&\Pr\left(\left\lvert\frac{1}{n}\sum_{i=1}^n\frac{\delta_i}{\pi^*(\bX_i)\widehat{\pi}(\bX_i)}\{\widehat{\pi}(\bX_i) - \pi^*(\bX_i)\}\right\rvert \geq Cn^{-1/3-\zeta/2}\right)\notag \\ 
		&\leq&Pr\left(\left\lvert\frac{1}{n}\sum_{i=1}^n\frac{2\delta_i}{\kappa^2}\{\widehat{\pi}(\bX_i) - \pi^*(\bX_i)\}\right\rvert \geq Cn^{-1/3-\zeta/2}\right) + P(A_n)
		\notag\\
		&\leq& \frac{\E[2\kappa^{-2}n^{-1}\sum_{i=1}^n[\delta_i\{\pi^*(\bX_i)\widehat{\pi}(\bX_i)\}^{-1}\{\widehat{\pi}(\bX_i) - \pi^*(\bX_i)\}]^2}{C^2n^{-2/3-\zeta}}\notag \\ 
		&&+ P(A_n)\notag \\ 
		&\leq&\frac{2\sum_{i=1}^n\E\{\widehat{\pi}(\bX_i) - \pi^*(\bX_i)\}^2}{n\kappa^2C^2n^{-2/3-\zeta}} + P(A_n)\notag \\ 
		&\leq&\frac{\E\{\widehat{\pi}(\bX_1) - \pi^*(\bX_1)\}^2}{\kappa^2C^2n^{-2/3-\zeta}} + P(A_n)\notag\\
		&=&o_{\mathbb{P}}(1),\label{eq: app sec -1/3}
	\end{eqnarray}
	where $\bX_i$ is the random variable associated with $\bx_i$, the first inequality holds by the Markov inequality, the second inequality holds by Assumption~4 and the fact that $\delta_i\leq 1$ for $i=1,\ldots,n$, the third inequality holds since $\bX_1,\ldots,\bX_n$ are identically distributed, and the fourth inequality holds by Lemma~\ref{cor: 1}. By (\ref{eq: app sec -1/3}),  we have validated (\ref{eq: A4}), so  we have 
	\begin{equation}
		\frac{1}{n}\sum_{i=1}^n\left\{1-\frac{\delta_i}{\hpii}\right\}\left\{\hmi - \mi\right\}=o_{\mathbb{P}}(n^{-1/2}).\label{eq: clt second}
	\end{equation}
	
	By (\ref{eq: CLT first}), (\ref{eq: clt third}) and (\ref{eq: clt second}), we have proved Theorem \ref{thm:clt} by the Slutsky's theorem \citep[Theorem~9.1.6]{athreya2006measure}.
\end{proof}



\bibliographystyle{dcu}
\bibliography{bibtex}


 





\end{document}